%% file: main.tex
\documentclass[11pt]{article}

\usepackage{amsmath,amssymb,amsfonts}
\usepackage{amsthm}
\usepackage{mathtools}
\usepackage{bm}
\usepackage{booktabs}
\usepackage{graphicx}
\usepackage{placeins}
\usepackage{float}
\usepackage{cite}
\usepackage[margin=1in]{geometry}
\usepackage[hidelinks]{hyperref}

\newtheorem{theorem}{Theorem}
\newtheorem{proposition}{Proposition}
\newtheorem{corollary}{Corollary}

\newcommand{\R}{\mathbb{R}}
\newcommand{\rank}{\operatorname{rank}}

\title{Beyond Observability: Sharp Sample Bounds for Minimax State Reconstruction}
\hypersetup{
  pdftitle={Beyond Observability: Sharp Sample Bounds for Minimax State Reconstruction},
  pdfkeywords={State estimation, observability, minimax reconstruction, bounded measurement noise, optimal sampling, linear systems}
}
\ifdefined\IJCAnonymous
\author{}
\hypersetup{pdfauthor={}}
\else
\author{Jerzy Baranowski\thanks{Corresponding author: \href{mailto:jb@agh.edu.pl}{jb@agh.edu.pl}. ORCID: \href{https://orcid.org/0000-0003-3313-581X}{0000-0003-3313-581X}.}\\
\small Department of Automatic Control and Robotics\\
\small AGH University of Krakow, Poland}
\hypersetup{pdfauthor={Jerzy Baranowski}}
\fi
\date{}

\begin{document}

\maketitle

\begin{abstract}
We study how many output samples match the worst-case state reconstruction error obtainable from a continuous history under bounded measurement noise. For observable scalar-output linear time-varying systems with $n$ states, $n$ samples suffice for exact reconstruction, whereas at most $n^2$ attain the continuous-history minimax error in the $\ell_\infty$ norm. An explicit construction proves sharpness even for real-analytic systems. For time-invariant systems, strict total positivity yields one $n$-sample schedule optimal for every state coordinate; endpoint-observed tridiagonal cooperative chains satisfy this conclusion in physical coordinates. Interior sensing behaves differently: an explicit three-state family requires four samples despite positivity and distinct real eigenvalues. The strict three-state gap persists under small stabilizing damping. A nilpotent four-state example and an interval proof for a six-state mechanical system give gaps. The results identify how temporal sampling and the prescribed state-error geometry determine the measurements required for minimax reconstruction. Time-domain experiments illustrate the resulting estimation bounds.
\end{abstract}


\input{submission/word-count}

\input{sections/01_introduction}
\input{sections/02_ltv_online}
\input{sections/03_minimax}
\input{sections/04_memory_budgets}
\input{sections/02b_history_compression}
\input{sections/02_finite_memory}
\input{sections/05_structured_lti}
\input{sections/06_counterexample}
\input{sections/07_building}
\input{sections/07b_numerical_comparison}
\input{sections/10_numerics}
\FloatBarrier
\input{sections/11_discussion}

\ifdefined\IJCAnonymous\else
\section*{Acknowledgment}
The author thanks OpenAI for including him in the ChatGPT for Academic Researchers program.

\section*{Funding}
This work was internally funded through the AGH Subvention for Scientific Research (internal application no.~18277).
\fi

\section*{Declaration of generative AI use}
The author used OpenAI ChatGPT during preparation of this article for editorial and organizational assistance, language refinement, and assistance in developing and checking mathematical and computational drafts. This assistance affected the presentation throughout the manuscript. All mathematical derivations, numerical results, citations, scientific claims, and final text were reviewed and validated by the author, who assumes full responsibility for the content.

\section*{Competing interests}
The author reports no competing interests.

\section*{Data availability statement}
No external datasets were used in this theoretical study. The models and parameters used in the numerical examples are specified in the article.
\FloatBarrier
\bibliographystyle{plainurl}
\ifdefined\IJCAnonymous
\bibliography{references,references_recent,references_pending_anonymous}
\else
\bibliography{references,references_recent,references_pending}
\fi

\appendix
\input{sections/appendix_history_compression}
\input{sections/appendix_lti_gaps}

\end{document}

%% file: submission/word-count.tex
\noindent\textbf{Word count:} approximately 7400, including references, appendices, and captions; excluding mathematical expressions.

%% file: sections/01_introduction.tex
\section{Introduction}
\label{sec:introduction}

State estimation from past input and output measurements is a basic task in control. When only a finite part of the output history is retained, the available measurements must contain enough information to determine the state. In the presence of measurement errors, they must also support a sufficiently accurate estimate. These two requirements lead to different questions about the number and timing of measurements.

Exact reconstruction from finitely many output samples has a long history. Gilchrist's $n$-observability formalized recovery from finite output histories \cite{gilchrist1966nobservability}, while finite-memory observers and deadbeat observers in continuous time express the estimate through delayed input and output values \cite{medvedev1992finite,medvedev1994deadbeat,medvedev1994feedforward}. Such constructions extend to linear time-varying systems and uncertain models \cite{graton2014fmo,graton2014sensitivity,medvedev2010sensitivity}. More recent work considers delayed measurements \cite{bobtsov2021delayedltv,cacace2022delays}, aperiodic sampling and jitter \cite{sferlazza2021jitter,huff2024aperiodic}, sporadic data with timestamp uncertainty \cite{guarro2024hybrid}, and direct reconstruction within observer architectures \cite{meiners2025direct}. Delayed-output observers have also been developed for nonlinear systems \cite{arezki2026delayed}.

The accuracy of a reconstruction depends on measurement timing as well as observability. Nonuniform sampling has been used to reduce absolute and relative measurement-error propagation \cite{troch1973sampling,dormido1979relative,delasen1981nonperiodic}. In stochastic estimation, measurement allocation balances the number and quality of observations \cite{aspeel2021budget}. Deterministic bounded-error and set-membership methods instead seek guarantees valid for every admissible disturbance \cite{rohou2021bounded,althoff2021guaranteed,mazenc2022sampled,meslem2025delayedset}. Optimal sampling under unknown-but-bounded errors was studied by Belforte, Bona, and Frediani \cite{belforte1987sampling}. Micchelli relates bounded-error parameter uncertainty to optimal recovery and derives extremal properties for strictly totally positive measurement matrices \cite{micchelli1989sampling}. Related worst-case experiment-design problems use radius-of-information criteria \cite{tanaskovic2014worstcase}. For polynomial regression, optimal worst-case designs also depend on the basis in which parameter uncertainty is measured \cite{belforte2002polynomial}. These approaches motivate studying the smallest estimation error attainable under a prescribed bound on measurement noise in fixed state coordinates.

For nonlinear systems, singular sampling can change the noise exponent in local minimax reconstruction \cite{singularSampling2026}. We consider linear systems, where the noise dependence remains linear, and study the number of samples required to attain continuous-history accuracy.

Finite measurement sets that reproduce a continuous design are also studied in optimal experimental design. For dynamic parameter-estimation problems, Jost, Sager, and Le construct finite grids equivalent to continuous Fisher-information designs \cite{jost2017measurementtimes}. Classical finite-atomic representation results provide support-reduction tools for such questions \cite{schmudgen2017moment}. In the present problem, each selected time supplies an actual output value. The same selection must recover all state coordinates, and convex combinations of information rows need not be realizable as individual measurements. The question concerns the cardinality of this common temporal selection rather than the support of one coordinate recovery functional.

This leaves a question beyond exact observability: how many measurements are needed to attain the estimation accuracy available from the complete continuous output history? For an observable scalar-output $n$-state system, $n$ suitably chosen measurements determine the state in the absence of noise. They need not attain the continuous-history minimax error under bounded noise. We distinguish the minimum number for exact reconstruction, $m_{\mathrm{obs}}$, from the minimum number attaining that error, $m_{\mathrm{sat}}$.

Our main result gives the sharp bound
\begin{equation}
    m_{\mathrm{obs}}=n\le m_{\mathrm{sat}}\le n^2
\end{equation}
for continuous scalar-output linear time-varying (LTV) systems. We prove sharpness by explicitly constructing, for every $n\ge2$, a real-analytic scalar-output LTV system requiring exactly $n^2$ distinct sampling times. Thus time regularity alone cannot improve the bound. For linear time-invariant (LTI) systems, we establish simultaneous coordinatewise attainment by $n$ samples under strict total positivity and for endpoint-observed tridiagonal cooperative chains in physical coordinates. In contrast, a three-state family with an interior measurement requires four samples over an explicit range of its parameters, and a strict gap persists under small stabilizing damping. An exact nilpotent $4\to5$ example and a computer-assisted mechanical $6\to7$ result further distinguish the role of coordinate geometry from observability. All results use specified state and noise norms and concern a fixed reconstruction time.

Time-domain experiments compare estimated trajectories with their error bounds and illustrate the effect of measurement timing and of an additional sample.

Section~\ref{sec:problem} formulates the estimation problem and defines the two measurement requirements. Section~\ref{sec:history-compression} proves the sharp quadratic bound. Section~\ref{sec:lti} treats time-invariant systems and the strict-gap examples. Section~\ref{sec:numerics} presents the numerical experiments, and Section~\ref{sec:discussion} discusses the scope and open questions. Detailed geometric, algebraic, and computer-assisted proofs are collected in the appendices.

%% file: sections/02_ltv_online.tex
\section{Problem formulation}
\label{sec:problem}

\subsection{System, measurements, and reconstruction}
\label{sec:ltv-online}

Consider the linear time-varying system
\begin{equation}
    \dot x(t)=A(t)x(t)+B(t)u(t),\qquad
    y(t)=C(t)x(t)+v(t),
\end{equation}
where $x(t)\in\mathbb R^n$, the output is scalar, and the system matrices are continuous. The input $u$ and the system matrices are known. Measurement noise satisfies $|v(s)|\le\delta$. For a fixed horizon $T>0$, the task at time $t$ is to estimate the current state $x(t)$ from measurements in the trailing window $[t-T,t]$ and the input history over that window.

Choose distinct relative delays $0\le\tau_1<\cdots<\tau_m\le T$. The available output values are $y(t-\tau_i)$, as illustrated in Figure~\ref{fig:finite-memory-window}. Their dependence on the known input can be removed. If $\Phi(t,s)$ denotes the state-transition matrix, define
\begin{equation}
    q_i[u;t]=C(t-\tau_i)\int_{t-\tau_i}^{t}
                  \Phi(t-\tau_i,s)B(s)u(s)\,ds,
    \qquad z_i(t)=y(t-\tau_i)+q_i[u;t].
\end{equation}
Variation of constants then gives
\begin{equation}
    z_i(t)=C(t-\tau_i)\Phi(t-\tau_i,t)x(t)+v(t-\tau_i).
\end{equation}
The correction vanishes when $u=0$.

Let $Z_\tau(t)$ and $V_\tau(t)$ collect the corrected measurements and their errors. With
\begin{equation}
    \mathcal D_\tau(t)=\operatorname{col}_i
           C(t-\tau_i)\Phi(t-\tau_i,t),
    \label{eq:ltv-delay-matrix}
\end{equation}
the reconstruction problem becomes
\begin{equation}
    Z_\tau(t)=\mathcal D_\tau(t)x(t)+V_\tau(t),
    \qquad\|V_\tau(t)\|_\infty\le\delta.
    \label{eq:ltv-delay-model}
\end{equation}
When $\mathcal D_\tau(t)$ has full column rank, a left inverse defines the estimator
\begin{equation}
    \widehat x(t)=L_\tau(t)Z_\tau(t),\qquad
    L_\tau(t)\mathcal D_\tau(t)=I.
\end{equation}
It reconstructs the state exactly without noise. With noise, its error is
\begin{equation}
    e(t):=\widehat x(t)-x(t)=L_\tau(t)V_\tau(t),\qquad
    \|e(t)\|_\infty\le\delta\|L_\tau(t)\|_{\infty\to\infty}.
    \label{eq:basic-estimation-bound}
\end{equation}
Thus the induced norm of the left inverse is the amplification factor for bounded measurement errors. We seek both an appropriate selection of measurement times and a left inverse with the smallest such factor.

\begin{figure}[t]
    \centering
    \includegraphics[width=0.82\linewidth]{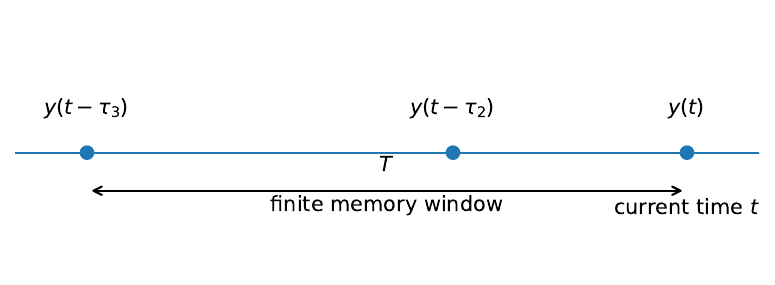}
    \caption{Measurements available for current-state reconstruction at time $t$. Selected output values lie in the window $[t-T,t]$; the known-input correction is applied before reconstruction.}
    \label{fig:finite-memory-window}
\end{figure}

All sample counts below refer to the measurements used in one reconstruction. Maintaining their availability as the window moves can require additional storage.

%% file: sections/03_minimax.tex
\subsection{The smallest worst-case error}
\label{sec:minimax}

Fix a reconstruction time and write $D=\mathcal D_\tau(t)$. For a given left inverse, the bound in \eqref{eq:basic-estimation-bound} is sharp over $\|V\|_\infty\le\delta$: noise values with signs matching a row of $L$ attain its row norm. Consequently the best linear estimator minimizes $\|L\|_{\infty\to\infty}$ subject to $LD=I$.

The choice of inverse matters when there are more measurements than state coordinates. Different generalized inverses can optimize different control objectives \cite{hunek2021theta,feliks2023generalized,hunekfeliks2025predictive}, and generalized inverses also occur in observer design \cite{meng2024generalizedobserver}. Here the objective is fixed by componentwise bounded measurement errors and the $\ell_\infty$ state-error norm.

Define the reconstruction margin
\begin{equation}
    \mu_\infty(D)=\min_{\|x\|_\infty=1}\|Dx\|_\infty.
    \label{eq:delay-margin}
\end{equation}
A small margin means that a unit change in the state can produce only a small change in the measurements. Bounded noise can then conceal a large state error.

\begin{proposition}[Minimax reconstruction error]
\label{prop:minimax-left-inverse}
For full-column-rank $D$,
\begin{equation}
    \gamma_\infty(D):=\min_{LD=I}\|L\|_{\infty\to\infty}
       =\frac{1}{\mu_\infty(D)},\qquad
    R^\star(D,\delta)=\delta\gamma_\infty(D).
    \label{eq:minimax-left-inverse}
\end{equation}
The worst-case radius $R^\star$ is also optimal among arbitrary estimators on the unbounded state space.
\end{proposition}

\begin{proof}
The rows of a left inverse can be optimized separately. For coordinate $j$, linear-programming duality gives
\begin{equation}
    \alpha_j(D):=\min_{D^Ta=e_j}\|a\|_1
      =\max_{\|Dx\|_\infty\le1}x_j.
    \label{eq:coordinate-duality}
\end{equation}
By symmetry of the feasible set,
\begin{equation}
    \max_j\alpha_j(D)
      =\max_{\|Dx\|_\infty\le1}\|x\|_\infty
      =\frac{1}{\mu_\infty(D)}.
    \label{eq:alpha-mu-equivalence}
\end{equation}
This proves the optimal linear gain. For every $x$ satisfying $\|Dx\|_\infty\le\delta$, the states $x$ and $-x$ can both produce the observation zero under admissible errors. No estimator can distinguish them, so their separation gives the same lower bound.
\end{proof}

The coordinate system and error norm remain fixed when measurement selections are compared. If physical variables require different scales, positive diagonal matrices $D_x,D_v$ give the normalized matrix $D_v^{-1}DD_x$; the result then applies in the corresponding weighted norms. We use the stated coordinates and the unweighted $\ell_\infty$ norm below.

%% file: sections/04_memory_budgets.tex
\subsection{Number of measurements for exact and minimax reconstruction}
\label{sec:memory-budgets}

At a fixed current time $t$, let
\begin{equation}
    h_t(\tau):=C(t-\tau)\Phi(t-\tau,t),\qquad 0\le\tau\le T.
\end{equation}
The complete continuous history has margin
\begin{equation}
    \mu_\infty^{\mathrm{cont}}(t,T)
    :=\min_{\|x\|_\infty=1}\sup_{\tau\in[0,T]}|h_t(\tau)x|.
    \label{eq:continuous-history-margin}
\end{equation}
For a finite schedule, write $\mu_\infty(\tau;t):=\mu_\infty(\mathcal D_\tau(t))$. We distinguish
\begin{align}
    m_{\mathrm{obs}}(t,T)
    &:=\min\{m:\exists\tau\in[0,T]^m,\ \rank\mathcal D_\tau(t)=n\},\\
    m_{\mathrm{sat}}(t,T)
    &:=\min\{m:\exists\tau\in[0,T]^m,\
       \mu_\infty(\tau;t)=\mu_\infty^{\mathrm{cont}}(t,T)\}.
\end{align}
The first number is the minimum for exact recovery without noise. The second is the minimum attaining the continuous-history minimax radius $\delta/\mu_\infty^{\mathrm{cont}}(t,T)$. We call this equality of finite and continuous-history accuracy \emph{saturation}. Each selected time contributes one actual output measurement.

We assume the continuous information rows span $\mathbb R^n$. Then $n$ rows can be selected as a basis, while fewer cannot have full rank. Thus $m_{\mathrm{obs}}=n$ and $m_{\mathrm{sat}}\ge n$. The next section determines how much larger the second number can be.

%% file: sections/02b_history_compression.tex
\section{The sharp quadratic sample bound}
\label{sec:history-compression}

The main result answers the sample-count question without imposing additional structure on the continuous temporal information curve. Its upper bound gives a finite replacement for the full history; its sharpness statement rules out a smaller dimension-only bound.

\begin{theorem}[Finite temporal-history saturation]
\label{thm:n2-saturation}
For every observable continuous scalar-output $n$-state LTV temporal information curve,
\begin{equation}
    \boxed{
    n=m_{\mathrm{obs}}(t,T)
    \le
    m_{\mathrm{sat}}(t,T)
    \le n^2.
    }
\end{equation}
Moreover, the quadratic dimension-only bound is sharp even for real-analytic systems: for every $n\ge2$ there exists a scalar-output LTV system with real-analytic $A(\cdot)$ and $C(\cdot)$ for which
\begin{equation}
    \boxed{
    m_{\mathrm{sat}}(t,T)=n^2.
    }
\end{equation}
\end{theorem}

Define the centrally symmetric temporal measurement hull
\begin{equation}
    Q_t
    :=
    \operatorname{conv}
    \{\pm h_t(\tau):\tau\in[0,T]\}
    \subset\mathbb R^n.
\end{equation}
Then
\begin{equation}
    \sup_{\tau\in[0,T]}|h_t(\tau)x|
    =h_{Q_t}(x),
\end{equation}
where $h_{Q_t}(x)=\sup_{q\in Q_t}q^Tx$ is the support function. Write $B_1^n=\{q:\|q\|_1\le1\}$ for the cross-polytope with vertices $\pm e_j$. The continuous margin is its largest scale contained in the measurement hull:
\begin{equation}
    \mu_{\infty}^{\mathrm{cont}}(t,T)
    =
    \max\left\{
    \rho:\ \rho B_1^n\subseteq Q_t
    \right\}.
    \label{eq:crosspolytope-inclusion}
\end{equation}
Write $\mu_c:=\mu_{\infty}^{\mathrm{cont}}(t,T)$. For each coordinate axis define the maximal radial point
\begin{equation}
    q_j:=\rho_j e_j,
    \qquad
    \rho_j:=\max\{\rho\ge0:\rho e_j\in Q_t\}.
\end{equation}
Because $\mu_c B_1^n\subseteq Q_t$, one has $\rho_j\ge\mu_c$. Each $q_j$ lies on the boundary of $Q_t$ and therefore belongs to an exposed face contained in an affine hyperplane of dimension at most $n-1$. Carath\'eodory reduction within that face represents $q_j$ by at most $n$ signed temporal atoms. Taking the union over $j=1,\ldots,n$ therefore uses at most $n^2$ actual temporal rows. The resulting symmetric sampled hull contains every $\pm q_j$ and hence, by convexity and $\rho_j\ge\mu_c$, every $\pm\mu_c e_j$. It therefore contains $\mu_c B_1^n$ and attains the continuous margin. This support mechanism is closely related to classical finite-atomic representation and Carath\'eodory--Richter--Tchakaloff ideas \cite{schmudgen2017moment}; the support-reduction principle itself is not a new convex-geometric fact. A complete proof is given in Appendix~\ref{app:history-compression}.

For sharpness, the construction builds $n$ exposed simplex faces of the temporal hull, one for each coordinate direction, with
\begin{equation}
    \mu_c e_j
\end{equation}
lying in the relative interior of the $j$th face. Each coordinate vertex then requires all $n$ temporal vertices of that face. The faces use disjoint temporal atom sets, forcing $n^2$ distinct actual sample times. An analytic curve can pass through all required vertices while avoiding the critical faces at every other time. Thus regularity alone does not reduce the sample requirement; a smaller bound must use additional system structure.

For $n=2$, the construction can be visualized using the four informative row vectors
\begin{equation}
    (c,s),\quad(c,-s),\quad(s,c),\quad(-s,c),
\end{equation}
where
\begin{equation}
    c=\cos\frac{\pi}{8},
    \qquad
    s=\sin\frac{\pi}{8}.
\end{equation}
Their signed convex hull is a regular octagon. The continuous margin is $c$, while $ce_1$ and $ce_2$ lie in the relative interiors of two different exposed edges. Both endpoints of both edges are necessary, giving
\begin{equation}
    m_{\mathrm{sat}}=4=n^2.
\end{equation}

A point of $Q_t$ need not be one realizable measurement: the admissible measurements are the atomic rows $h_t(\tau)$. Sharpness depends on how many such rows are needed to represent the boundary points.

Theorem~\ref{thm:n2-saturation} is pointwise in the current time: the saturating delay set may change with $t$. It does not give a fixed pattern of size $n^2$ that is optimal for all $t$.

%% file: sections/02_finite_memory.tex
\section{Time-invariant systems}
\label{sec:lti}

Time-invariant dynamics restrict the form of the output history. We show that strict total positivity permits the continuous-history accuracy to be attained with $n$ measurements, and then give two systems for which an additional measurement is necessary.

\subsection{The time-invariant case}
\label{sec:finite-memory}

For constant matrices in Section~\ref{sec:problem}, $\Phi(t-\tau,t)=e^{-A\tau}$ and
\begin{equation}
    q_i[u;t]=C\int_{t-\tau_i}^{t}e^{A(t-\tau_i-s)}Bu(s)\,ds,
    \qquad z_i(t)=y(t-\tau_i)+q_i[u;t].
\end{equation}
Consequently
\begin{equation}
    Z_\tau(t)=\mathcal D_\tau x(t)+V_\tau(t),
    \qquad
    \mathcal D_\tau=\operatorname{col}_i Ce^{-A\tau_i}.
    \label{eq:delay-information-model}
\end{equation}
The correction vanishes for an autonomous system. The matrix and its margin are now independent of current time. We use the notation
\begin{equation}
    \mu_\infty^-(\tau):=\mu_\infty(\mathcal D_\tau),\qquad
    \mu_{\infty,\mathrm{delay}}^{\mathrm{cont}}(T)
    :=\min_{\|x\|_\infty=1}\sup_{\tau\in[0,T]}|Ce^{-A\tau}x|.
    \label{eq:continuous-delay-margin}
\end{equation}
The superscript $-$ distinguishes delayed current-state reconstruction from a forward information family $Ce^{As}$. We write $\gamma_\infty(\tau)=1/\mu_\infty^-(\tau)$ and $\alpha_j^-(\tau)=\alpha_j(\mathcal D_\tau)$. The two measurement counts are also independent of $t$.

\paragraph{Time orientation.}

Reconstruction at the beginning and end of the window has different noise gains. Let
\begin{equation}
    s_i=T-\tau_i,
\end{equation}
and define
\begin{equation}
    \mathcal O_s
    :=
    \begin{bmatrix}
        Ce^{As_1}\\
        \vdots\\
        Ce^{As_m}
    \end{bmatrix}.
\end{equation}
After subtraction of the known forced response, $\mathcal O_s$ reconstructs $x(t-T)$. The current-state matrix is related to it by
\begin{equation}
    \boxed{\mathcal D_\tau=\mathcal O_s e^{-AT}.}
    \label{eq:orientation-relation}
\end{equation}
Consequently the two matrices have the same rank. If $L_s\mathcal O_s=I$, current-state prediction has measurement gain
\begin{equation}
    K_\tau=e^{AT}L_s.
\end{equation}
Thus exact observability is invariant under the change of time orientation, whereas the norm of the reconstruction operator generally is not.

For a forward information example on $[a,b]$, the pair $\bar A=-A$, $\bar C=Ce^{Aa}$ gives the same rows at delays $\tau\in[0,b-a]$. The four-state example below uses this realization. For the mechanical example, a norm-preserving time-reversal symmetry instead transfers the result to delays of the same plant.

%% file: sections/05_structured_lti.tex
\subsection{Saturation at \texorpdfstring{$n$}{n} under strict total positivity}
\label{sec:structured-lti}

Strict total positivity provides a sufficient condition for $n$ measurements to attain the continuous-history minimax accuracy.

Let
\begin{equation}
    G(\tau)=
    \begin{bmatrix}
        g_1(\tau)&\cdots&g_n(\tau)
    \end{bmatrix}
    =Ce^{-A\tau},
    \qquad \tau\in[a,b].
\end{equation}
Write $\alpha_j^{-,\mathrm{cont}}=\sup\{|x_j|:\sup_{\tau\in[a,b]}|G(\tau)x|\le1\}$ for the continuous coordinate recovery constants. Assume the ordered basis $(g_1,\ldots,g_n)$ is strictly totally positive in the collocation sense. Thus every ordered minor of every collocation matrix formed from increasing delays and increasing column indices is positive. We use standard facts from Chebyshev systems and total positivity; see, e.g., \cite{karlin1966tchebycheff,karlin1968totalpositivity,pinkus2010totallypositive}.

\begin{theorem}[Simultaneous coordinatewise saturation]
\label{thm:stp-saturation}
Suppose $G$ is strictly totally positive on $[a,b]$. Then one $n$-point schedule $\tau^\star$ attains every continuous coordinate recovery constant:
\begin{equation}
    \alpha_j^-(\tau^\star)=\alpha_j^{-,\mathrm{cont}},
    \qquad j=1,\ldots,n.
    \label{eq:simultaneous-coordinate-saturation}
\end{equation}
In particular,
\begin{equation}
    m_{\mathrm{sat}}=m_{\mathrm{obs}}=n.
\end{equation}
The same conclusion holds after a signed permutation of state coordinates. With positive diagonal state scalings, the corresponding statement holds in the induced weighted $\ell_\infty$ norm.
\end{theorem}

\begin{proof}
Fix a coordinate $j$ and consider
\begin{equation}
    \rho_j
    :=
    \min_{c_k,\,k\ne j}
    \left\|
        g_j+
        \sum_{k\ne j}c_k g_k
    \right\|_\infty.
\end{equation}
Deleting one basis function from a strictly totally positive basis leaves a Haar space of dimension $n-1$. The Chebyshev alternation theorem therefore gives a unique minimax residual
\begin{equation}
    r_j(\tau)=G(\tau)z^{(j)},
    \qquad z_j^{(j)}=1,
\end{equation}
with $n$ alternating contacts
\begin{equation}
    r_j(t_i)
    =
    \sigma(-1)^{i-1}\rho_j,
    \qquad i=1,\ldots,n.
\end{equation}
After normalization by $\rho_j$, this residual is dual feasible and yields
\begin{equation}
    \alpha_j^{\mathrm{cont}}
    =
    \frac{1}{\rho_j}.
\end{equation}
Choose a critical coordinate $j_\star$ maximizing $\alpha_j^{\mathrm{cont}}$, and form the square collocation matrix $O$ at its alternating contacts. Strict total positivity implies the checkerboard sign pattern
\begin{equation}
    (-1)^{i+j}(O^{-1})_{ji}>0.
\end{equation}
Let $s_i=\sigma(-1)^{i-1}$. Since the normalized critical dual vector satisfies $O\widehat z_\star=s$, the checkerboard signs align with $s$ and give
\begin{equation}
    \|\operatorname{row}_k(O^{-1})\|_1
    =
    |(\widehat z_\star)_k|
    \le
    \alpha_k^{\mathrm{cont}}
    \le
    \alpha_\star.
\end{equation}
Restriction to finitely many times also gives $\alpha_k^{\mathrm{cont}}\le\|\operatorname{row}_k(O^{-1})\|_1$. Thus equality holds for every coordinate, proving \eqref{eq:simultaneous-coordinate-saturation}. Hence the $n$ contact delays attain the continuous minimax margin. Since scalar output requires at least $n$ samples for full column rank, $m_{\mathrm{sat}}=m_{\mathrm{obs}}=n$.
\end{proof}

Strict total positivity supplies the alternating contacts through classical best uniform approximation. Together with the checkerboard inverse signs, these contacts saturate the bounded-error reconstruction criterion for all state coordinates.

\paragraph{Real modal systems.}

For $A=\operatorname{diag}(a_1,\ldots,a_n)$ with distinct real $a_j$ and nonzero output coefficients $c_j$, the delayed basis $g_j(\tau)=c_j e^{-a_j\tau}$ inherits the total positivity of the exponential kernel \cite{karlin1968totalpositivity,pinkus2010totallypositive}. After a signed permutation of coordinates, Theorem~\ref{thm:stp-saturation} applies; the RC chain below illustrates the same measurement reduction directly in physical coordinates.

\paragraph{Endpoint measurements in a positive chain.}
The conclusion also holds in the physical coordinates of a nearest-neighbour RC or thermal chain; diagonalization of the plant is unnecessary.

\begin{corollary}[Endpoint-observed tridiagonal systems]
\label{cor:endpoint-chain}
Let $A$ be tridiagonal with $a_{i,i+1}>0$ and $a_{i+1,i}>0$, and let $C=e_1^T$ or $C=e_n^T$. For every $T>0$, an $n$-delay schedule on $[0,T]$ satisfies \eqref{eq:simultaneous-coordinate-saturation}. Consequently $m_{\mathrm{obs}}=m_{\mathrm{sat}}=n$. No restriction on the diagonal entries is required.
\end{corollary}

The proof in Appendix~\ref{app:endpoint-chain} uses the classical total positivity of tridiagonal cooperative flows \cite{margaliot2019tpds}, an endpoint Krylov identity, and alternation. It includes the possible contact at zero delay, where some minors vanish. For delayed current-state reconstruction the sign change $P=\operatorname{diag}(1,-1,1,\ldots)$ makes $-PAP$ cooperative and preserves the state norm. Thus the conclusion concerns the delayed rows of the same plant.

For an RC chain with node capacitances $c_i>0$ and neighbouring conductances $k_i>0$, the couplings are $a_{i,i+1}=k_i/c_i$ and $a_{i+1,i}=k_i/c_{i+1}$. The corollary therefore applies directly to endpoint sensing. The stronger coordinatewise equality means that the same schedule is optimal for every positive diagonal weighting of the state error.

%% file: sections/06_counterexample.tex
\subsection{Strict gaps beyond total positivity}

Interior sensing can require more than $n$ samples even in a cooperative chain with distinct real eigenvalues. We first give a three-state family, then a nilpotent example and a mechanical system. Their proofs are in Appendix~\ref{app:lti-gaps}.

\paragraph{A three-state family with an interior measurement.}
Let $\omega,w>0$ and consider
\begin{equation}
    A_w=\omega\begin{bmatrix}
    -1/2&1/2&0\\
    3/4&-1&1/(4w)\\
    0&3w/2&-3/2
    \end{bmatrix},
    \qquad C=e_2^T.
    \label{eq:interior-chain-model}
\end{equation}
Its eigenvalues are $0,-\omega,-2\omega$, and the middle output is observable. Set $R=e^{\omega T}$ and define
\begin{equation}
\begin{aligned}
    a(R)&=\frac{3R^2+42R+19}{3(R-1)^2},\\
    c(R)&=\frac{R^2-6R+21}{R^2-6R-3},
    &b(R)&=\frac{R^2+2R+29}{(R-1)^2}.
\end{aligned}
\end{equation}

\begin{theorem}[A positive three-state family requiring four samples]
\label{thm:interior-gap}
If $R>9$ and
\begin{equation}
    \frac{a(R)}{b(R)}<w\le\frac{a(R)}{c(R)},
    \label{eq:interior-gap-range}
\end{equation}
then the delayed current-state problem for \eqref{eq:interior-chain-model} satisfies
\begin{equation}
    m_{\mathrm{obs}}=3,\qquad m_{\mathrm{sat}}=4,
    \qquad\gamma_\infty^{\mathrm{cont}}=a(R).
\end{equation}
A saturating delay set is
\begin{equation}
    \left\{0,\ \frac1\omega\log\frac{3(R+1)}{R-3},\
    \frac1\omega\log\frac{R+1}{2},\ T\right\}.
\end{equation}
The parameter interval in \eqref{eq:interior-gap-range} is nonempty for every $R>9$.
\end{theorem}

After a norm-preserving column sign change and the variable $u=e^{\omega\tau}$, the information row is quadratic in $u$. Its first coordinate has a dual optimum with precisely three contacts, forcing the nodes $1,(R+1)/2,R$ into every saturating schedule. At these three nodes the third reconstruction row has norm $wb(R)>a(R)$, although its continuous optimum is $wc(R)\le a(R)$. A fourth node attains that optimum without changing the first coordinate's support. Appendix~\ref{app:interior-gap} proves these statements exactly.

For example, $R=16$ and $w=7/4$ give
\begin{equation}
    \gamma_\infty^{\mathrm{cont}}=\frac{1459}{675},
    \qquad
    \alpha_3^{-,\mathrm{cont}}=\frac{1267}{628},
    \qquad
    \alpha_3^-(\tau_{\mathrm{mandatory}})=\frac{2219}{900}
    >\frac{1459}{675}.
\end{equation}
The matrix $A_w$ is a positive diagonal similarity of $A_1$, a conservative RC-chain matrix. Varying $w$ therefore changes the coordinate scale relevant to the prescribed state-error norm while leaving the eigenvalues unchanged. Positivity and a real modal spectrum alone do not give the coordinatewise conclusion of Corollary~\ref{cor:endpoint-chain}.

\begin{corollary}[A strict gap in Hurwitz cooperative systems]
\label{cor:hurwitz-interior-gap}
For every parameter choice in Theorem~\ref{thm:interior-gap}, there exists $\varepsilon_0>0$ such that replacing $A_w$ by $A_w-\varepsilon I$, with $0<\varepsilon<\varepsilon_0$, gives a Hurwitz tridiagonal cooperative system with $m_{\mathrm{obs}}=3$ and $m_{\mathrm{sat}}>3$ on the same delay interval.
\end{corollary}

The proof uses continuity of the optimized three-sample margin and of the continuous-history margin. This persistence result establishes a strict gap; it does not determine the exact saturation count of the damped family.

\paragraph{An exact four-state gap.}
\label{sec:exact-counterexample}

The extended complete Chebyshev (ECT) property of the output space alone does not imply $n$-sample saturation. Consider
\begin{equation}
    H(t)=\begin{bmatrix}1&t&t^2&t^3\end{bmatrix}M,
    \qquad t\in[-1,1],
    \qquad
    M=\begin{bmatrix}
    -3&-1&1&-2\\
    3&-3&0&1\\
    3&0&-2&1\\
    2&2&-1&-1
    \end{bmatrix}.
    \label{eq:polynomial-information-row}
\end{equation}
Since $\det M=40$, its scalar trajectory space is $\mathcal P_3$. Let $J$ be the nilpotent Jordan block with $J_{i,i+1}=1$, and set
\begin{equation}
    S=\operatorname{diag}(1,1,2,6)M,
    \qquad A=S^{-1}JS,
    \qquad C=e_1^TS.
\end{equation}
Then $Ce^{At}=H(t)$. The causal delay realization
\begin{equation}
    \bar A=-A,\qquad \bar C=Ce^{-A}
\end{equation}
satisfies
\begin{equation}
    \bar C e^{-\bar A\tau}=H(\tau-1),
    \qquad \tau\in[0,2].
    \label{eq:counterexample-delay-shift}
\end{equation}

\begin{proposition}[Exact $4\to5$ gap]
\label{prop:four-five-gap}
For this four-state scalar-output LTI system,
\begin{equation}
    m_{\mathrm{obs}}=4,
    \qquad m_{\mathrm{sat}}=5,
    \qquad
    \frac{1}{\mu_{\infty,\mathrm{delay}}^{\mathrm{cont}}(2)}
    =\frac{14101}{11400}.
\end{equation}
A saturating delay schedule is
\begin{equation}
    \tau_5=\left(0,\frac25,\frac{19}{17},\frac85,2\right).
    \label{eq:counterexample-five-delay}
\end{equation}
\end{proposition}

Appendix~\ref{app:exact-gap-proof} gives an exact left inverse and a matching dual lower bound. The third coordinate is uniquely critical, and every saturating schedule must contain the delays $0$, $19/17$, and $2$. With these contacts fixed, the coordinate-$2$ recovery norm satisfies $\rho_2>3/2>14101/11400$ for every fourth delay. Appendix~\ref{app:exact-gap-proof} proves this by polynomial positivity, excluding four-sample saturation. This bound applies under the mandatory-contact restriction, not to the unrestricted four-sample optimum.

The example has analytic, nonoscillatory dynamics and an ECT output space. Its gap arises from the relation between the chosen state norm and the temporal basis; those properties of the output space do not supply the sign geometry used in Theorem~\ref{thm:stp-saturation}.

%% file: sections/07_building.tex
\paragraph{A six-state mechanical system.}
\label{sec:building}

Consider the three-storey undamped building in Figure~\ref{fig:building-model}
\begin{equation}
    \ddot q+Kq=0,
    \qquad
    K=\begin{bmatrix}
    2&-1&0\\
    -1&2&-1\\
    0&-1&1
    \end{bmatrix},
\end{equation}
with state $x=(q_1,q_2,q_3,\dot q_1,\dot q_2,\dot q_3)$ and output $y=q_1$. Thus
\begin{equation}
    A=\begin{bmatrix}0&I\\-K&0\end{bmatrix},
    \qquad C=\begin{bmatrix}1&0&0&0&0&0\end{bmatrix}.
\end{equation}
\begin{figure}[t]
    \centering
    \includegraphics{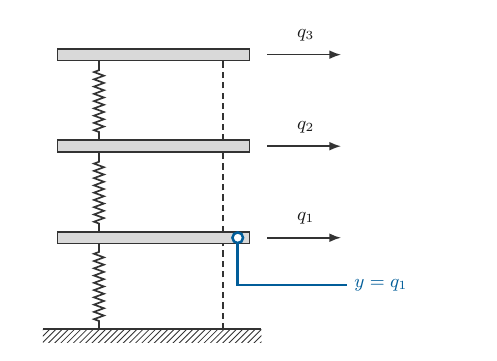}
    \caption{Three-storey lumped-mass model. The coordinates $q_1,q_2,q_3$ are horizontal floor displacements relative to the fixed foundation. Spring symbols denote the lateral couplings represented by $K$; the measured displacement is $y=q_1$. The drawing is schematic.}
    \label{fig:building-model}
\end{figure}

We use $T=10$ and the stated state coordinates. With $S=\operatorname{diag}(I_3,-I_3)$, the identities $SAS=-A$ and $CS=C$ give
\begin{equation}
    Ce^{-A\tau}=Ce^{A\tau}S.
    \label{eq:building-time-reversal}
\end{equation}
Since $S$ preserves the $\ell_\infty$ norm, forward-window and current-state delay margins coincide for every schedule of this same plant.

\begin{proposition}[$6\to7$ gap]
\label{prop:building-gap}
For the building above on $[0,10]$,
\begin{equation}
    m_{\mathrm{obs}}=6,\qquad m_{\mathrm{sat}}=7.
\end{equation}
The third coordinate is uniquely critical, with continuous recovery constant
\begin{equation}
    \alpha_\star\in[2.79547570325,\,2.79547570327].
    \label{eq:building-alpha}
\end{equation}
If $\tau_{c,1},\tau_{c,2},\tau_{c,3}$ denote its exact active delays, a saturating schedule is
\begin{equation}
    \tau_7=(1,\tau_{c,1},2.5,4,\tau_{c,2},7,\tau_{c,3}).
    \label{eq:building-seven-schedule}
\end{equation}
\end{proposition}

The proof is computer-assisted; its interval inclusions and inequalities are given in Appendix~\ref{app:building-proof}. The argument has two parts. First, primal--dual equality and uniqueness of the critical representation force every saturating schedule to contain three delays, approximately $1.848608$, $5.620512$, and $7.986152$. A six-sample schedule can then add only three free delays. A finite interval subdivision excludes saturation throughout their ordered domain. Second, interval bounds on the reconstruction rows show that \eqref{eq:building-seven-schedule} attains the continuous optimum. The quantitative lower bound $2.8056>\alpha_\star$ obtained in the exclusion step applies to schedules containing the three mandatory contacts; the unrestricted six-sample conclusion is strict nonattainment of the optimum.

The gap therefore occurs in a mechanical system with direct displacement sensing, as well as in the polynomial construction of Section~\ref{sec:exact-counterexample}.

%% file: sections/07b_numerical_comparison.tex
\subsection{Measurement requirements in the examples}
\label{sec:schedule-comparison}

Table~\ref{tab:memory-hierarchy} summarizes the measurement requirements. We give numerical reconstruction gains for the examples to quantify the difference between exact observability and minimax accuracy.

\begin{table}[t]
\centering
\caption{Numbers of measurements required for exact reconstruction and continuous-history minimax accuracy.}
\label{tab:memory-hierarchy}
\begin{tabular}{lccc}
\toprule
Class / example & $n$ & $m_{\mathrm{obs}}$ & $m_{\mathrm{sat}}$ \\
\midrule
General scalar-output LTV (analytic sharpness) & $n$ & $n$ & $\le n^2$ (sharp) \\
Strictly totally positive scalar LTI & $n$ & $n$ & $n$ \\
Endpoint-observed cooperative chain & $n$ & $n$ & $n$ \\
Interior-sensing family, \eqref{eq:interior-gap-range} & $3$ & $3$ & $4$ \\
Exact nilpotent ECT example & $4$ & $4$ & $5$ \\
Three-storey building & $6$ & $6$ & $7$ \\
\bottomrule
\end{tabular}
\end{table}

\paragraph{Reconstruction accuracy with fewer measurements.}
Table~\ref{tab:numerical-memory-gaps} compares saturating schedules with lower-cardinality schedules found numerically. Global optimality of the latter schedules has not been proved.

\begin{table}[t]
\centering
\caption{Numerical reconstruction gains for the strict-gap examples. Smaller $\gamma_\infty$ is better.}
\label{tab:numerical-memory-gaps}
\begin{tabular}{lccc}
\toprule
Example & Schedule & $m$ & $\gamma_\infty$ \\
\midrule
Polynomial & best found & $4$ & $1.303034$ \\
Polynomial & saturating & $5$ & $1.236930$ \\
Building & best found & $6$ & $2.864309$ \\
Building & saturating & $7$ & $2.795476$ \\
\bottomrule
\end{tabular}
\end{table}

For the four-state example, repeated unrestricted searches yielded endpoints $0$ and $2$. Optimization of the two interior delays gives
\begin{equation}
    (0,\,0.58708375,\,1.24583768,\,2),
    \label{eq:counterexample-four-delay}
\end{equation}
with gain $1.30303404$, above the continuous value $\alpha_\star=14101/11400\approx1.23692982$. The proof's bound $\rho_2>3/2$ applies only after imposing the three contacts required for saturation; it is not a bound on the unrestricted four-delay optimum.

For the building, an unrestricted differential-evolution search over six delays found
\begin{equation}
\begin{split}
    \tau_6^{\mathrm{num}}
    \approx{}&(1.02001,2.11812,4.25925,\\
    &5.63353,6.78593,8.02749),
\end{split}
\end{equation}
with gain $2.86431$, compared with $2.7954757$ for the saturating seven-delay schedule. The interval proof excludes saturation by any six-delay schedule. Its quantitative bound $\gamma_\infty>2.8056$ applies to schedules containing the three mandatory critical contacts, not to the unrestricted six-delay problem.

%% file: sections/10_numerics.tex
\section{Numerical examples}
\label{sec:numerics}
\label{sec:robustness}

We compare measurement schedules in a seven-state RC chain and in the four-state polynomial system of Section~\ref{sec:lti}. The RC chain illustrates the effect of timing and window length when seven measurements suffice for continuous-history accuracy. The polynomial system shows the improvement obtained by adding a fifth measurement.

\subsection{Models and simulation conditions}

Figure~\ref{fig:rc7-ladder} shows the RC chain. Its state $x=(v_1,\ldots,v_7)^T$ consists of the nodal voltages, and $y=v_1$ is the endpoint measurement. The capacitances $c_i$ and conductances $k_i$ are one in normalized quantities. With both ends resistively connected to ground, the voltage equations are
\begin{equation}
    c_i\dot v_i=k_{i-1}(v_{i-1}-v_i)+k_i(v_{i+1}-v_i),
    \qquad v_0=v_8=0.
\end{equation}
Thus
\begin{equation}
    A=\operatorname{tridiag}(1,-2,1),\qquad C=e_1^T,
    \label{eq:rc7-plant}
\end{equation}
and all eigenvalues of $A$ are negative. Corollary~\ref{cor:endpoint-chain} gives $m_{\mathrm{obs}}=m_{\mathrm{sat}}=7$ in the nodal infinity norm.

\begin{figure}[t]
    \centering
    \includegraphics[width=\linewidth]{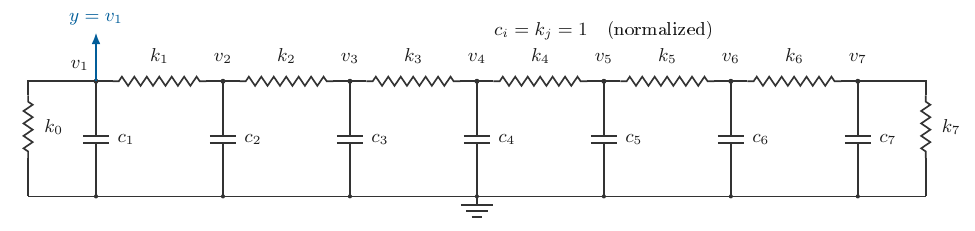}
    \caption{Seven-node RC chain with endpoint voltage measurement $y=v_1$. The capacitances are $c_1,\ldots,c_7$; the resistor labels $k_0,\ldots,k_7$ denote conductances. Both endpoint resistors connect to ground. Unit capacitances and conductances in normalized quantities give \eqref{eq:rc7-plant}.}
    \label{fig:rc7-ladder}
\end{figure}

Both systems are autonomous, with $u=0$. The RC initial voltages are
\begin{equation}
    x_0=(1,0.875,0.75,0.625,0.5,0.375,0.25)^T;
\end{equation}
the polynomial system uses $x_0=(1,-0.5,0.25,-0.125)^T$. The true state and each delayed noise-free output are evaluated using the matrix exponential. Reconstruction starts after the window is filled and uses 2001 equally spaced times in $[T,5T]$, with $T=6$ for the RC chain and $T=2$ for the polynomial system.

The measurement-noise bound is $\delta=10^{-3}$. Within each system, all schedules receive delayed values of the same scalar noise signal, denoted by $\nu$ here to distinguish it from the voltages. We use a sinusoid
\begin{equation}
    \nu_s(t)=\delta\sin(4\pi t/T+0.3)
\end{equation}
and a broader-band signal
\begin{equation}
    \nu_b(t)=\frac{\delta}{16}\sum_{\ell=1}^{16}\sin(\omega_\ell t+\phi_\ell).
\end{equation}
The angular frequencies are logarithmically spaced from $2\pi/T$ to $200\pi/T$. The phases are drawn uniformly from $[0,2\pi)$ with a fixed random seed of 41 and then held fixed. Both signals satisfy $|\nu(t)|\le\delta$; the broader-band signal is a finite sinusoidal sum.

For square information matrices we use their inverses. The RC contacts and weights are formed with 60-digit arithmetic before the trajectory evaluations in standard floating-point arithmetic; the longer-window checks below use 90 digits. The five-sample estimator uses the exact optimal left inverse in Appendix~\ref{app:exact-gap-proof}. Each reconstruction is compared with
\begin{equation}
    \|\widehat x(t)-x(t)\|_\infty\le\delta\gamma_\infty(\tau).
    \label{eq:simulation-error-bound}
\end{equation}

To check sharpness separately, choose a row $j$ of $L_\tau$ with the largest $\ell_1$ norm and prescribe
\begin{equation}
    \nu(t_\star-\tau_i)=\delta\operatorname{sign}(L_{ji}),\qquad t_\star=3T.
\end{equation}
Linear interpolation between these distinct measurement times gives a continuous bounded scalar noise signal. At $t_\star$, its coordinate-$j$ error equals $\delta\|\operatorname{row}_j L_\tau\|_1$. This signal depends on the estimator; the figures mark its error at that single time separately from the two common-noise trajectories.

\subsection{Timing and window length in the RC chain}

\begin{figure}[H]
    \centering
    \includegraphics[width=\linewidth]{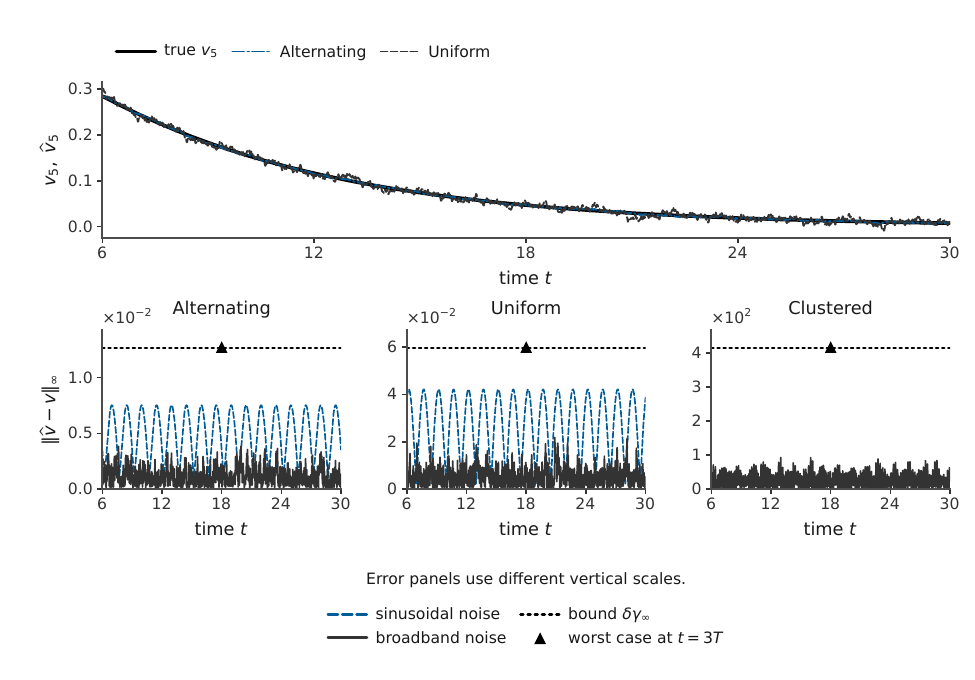}
    \caption{Current-state reconstruction in the seven-node RC chain \eqref{eq:rc7-plant}, with $u=0$, $T=6$, $\delta=10^{-3}$, and the initial voltages specified above. Top: the fifth voltage and its alternating and uniform estimates under the common broader-band noise. Lower panels: infinity-norm errors under a bounded sinusoid (dashed) and a sixteen-sinusoid signal (solid), with different vertical scales. Dotted lines give $\delta\gamma_\infty$; triangles mark attainment at $t=18$ by a separately chosen bounded scalar noise history for each estimator. The three seven-measurement schedules are specified in the text. Time and amplitudes are normalized.}
    \label{fig:ijc-rc7-trajectories}
\end{figure}

At $T=6$, the alternating contacts give the seven-delay schedule
\begin{equation}
\begin{split}
    \tau_7^{\mathrm{alt}}\approx{}&(0,2.29350070,3.99909195,4.98585012,\\
    &5.57673158,5.89772672,6).
\end{split}
    \label{eq:rc7-delays}
\end{equation}
The uniform delays are $(0,1,\ldots,6)$, and the clustered delays are $(0,0.2,\ldots,1.2)$. All three schedules have full rank. Their gains are approximately $12.6894$, $59.7770$, and $4.16327\times10^5$. The alternating schedule therefore reduces the guaranteed error by a factor of $4.71$ relative to the uniform schedule. The clustered case illustrates the effect of concentrating the measurements in a short part of the window.

Figure~\ref{fig:ijc-rc7-trajectories} shows the fifth voltage, which sets the largest coordinate recovery gain for \eqref{eq:rc7-delays}, and the reconstruction errors. All sampled errors lie below their bounds. Under sinusoidal noise, the maximum errors are $0.007529$, $0.04217$, and $0.26197$; under broader-band noise they are $0.003829$, $0.02181$, and $92.793$. The respective bounds are $0.012689$, $0.059777$, and $416.327$. Separately chosen disturbances attain each bound at $t_\star=18$ to numerical precision. The large amplification in the clustered case occurs despite exact noise-free observability.

Window length affects both the gain and the location of the retained measurements. Figure~\ref{fig:ijc-rc7-horizon} compares $T=1,\ldots,10$, with additional alternating schedules at $T=12,16,20$. The numerical transition occurs at $S_*\approx9.18905$. In the computed windows above this value, the forward alternating contacts $s_i^*$ are unchanged, while their current-state delays are $\tau_i(T)=T-s_i^*$. The smallest and largest delays are therefore $T-S_*$ and $T$, and their difference remains $S_*$. This constant span does not imply a finite maximum useful delay.

\begin{figure}[H]
    \centering
    \includegraphics[width=\linewidth]{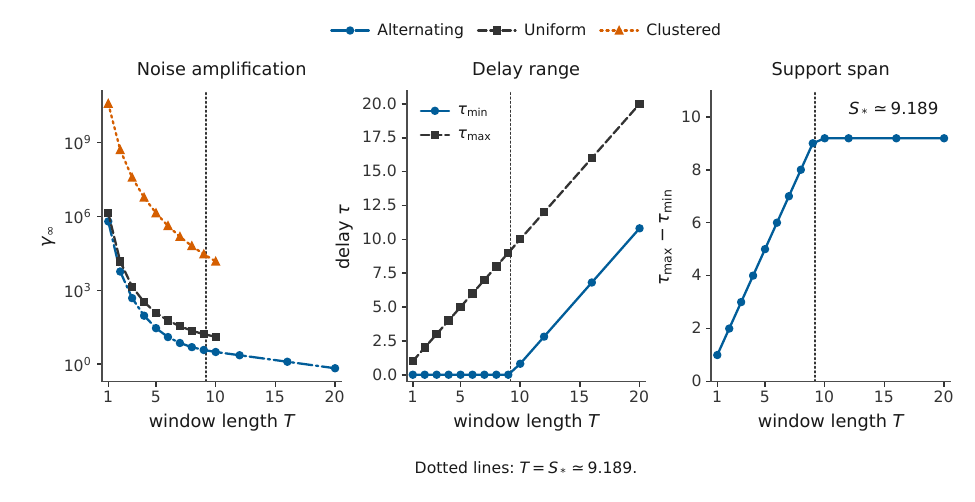}
    \caption{History-window length in the seven-node RC chain. Left: nodal noise gain on a logarithmic scale; uniform and clustered schedules end at $T=10$, while the alternating schedule includes $T=12,16,20$. The clustered delays span $[0,T/5]$. Middle: smallest and largest alternating delays. Right: their difference. Dotted vertical lines mark the numerical transition $S_*\approx9.18905$. The long-window support translates toward the oldest part of the history with a fixed span, while its gain continues to decrease. Markers denote computed windows; connecting lines are not fitted curves. All times are normalized.}
    \label{fig:ijc-rc7-horizon}
\end{figure}

The distinction follows from \eqref{eq:orientation-relation}. If $L_{s^*}\mathcal O_{s^*}=I$, the same forward measurements estimate the current state with
\begin{equation}
    K_\tau(T)=e^{AT}L_{s^*}.
\end{equation}
For any fixed full-rank forward schedule, this gain tends to zero as $T\to\infty$ because $A$ is Hurwitz. A positive plateau in the optimal current-state gain therefore cannot persist. For the computed alternating schedules, $\gamma_\infty$ decreases from $3.14367$ at $T=10$ to $0.68656$ at $T=20$, although the support span is unchanged. This conclusion uses exact dynamics and measurement noise only; the sampling count remains seven.

\subsection{Four versus five measurements}

For the causal polynomial realization \eqref{eq:counterexample-delay-shift}, compare the four-sample schedule \eqref{eq:counterexample-four-delay} with the saturating five-sample schedule \eqref{eq:counterexample-five-delay}. The respective bounds are $1.303034\times10^{-3}$ and $1.236930\times10^{-3}$. The additional sample therefore reduces the guaranteed error by approximately $5.1\%$ for these two schedules.

\begin{figure}[H]
    \centering
    \includegraphics[width=\linewidth]{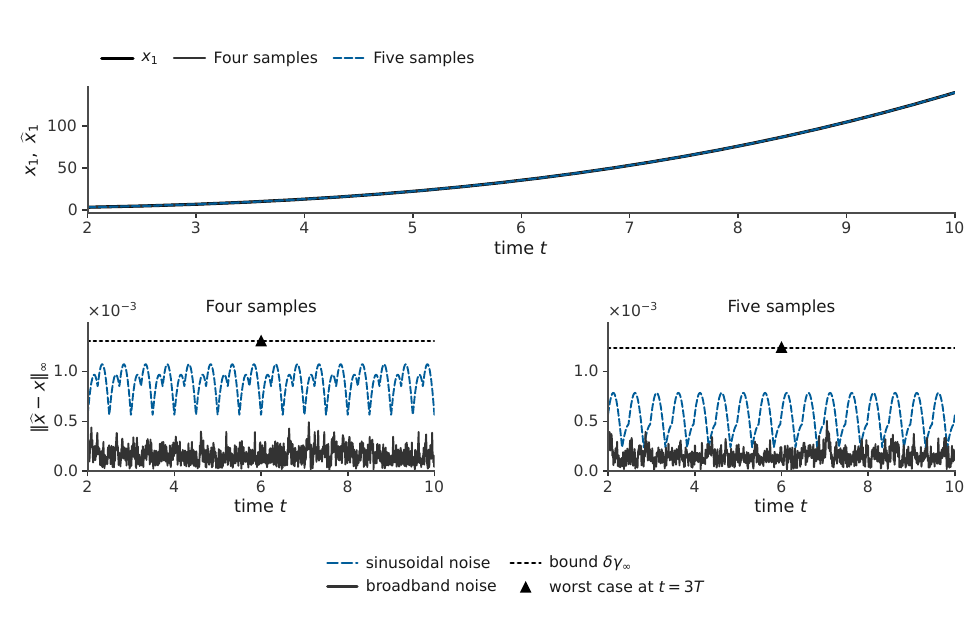}
    \caption{Reconstruction for the four-state polynomial system \eqref{eq:counterexample-delay-shift}, with $u=0$, $x_0=(1,-0.5,0.25,-0.125)^T$, $T=2$, and $\delta=10^{-3}$. Top: first state component and its four- and five-sample estimates under the common broader-band noise. Lower panels share the same vertical scale and show error norms for the two bounded noise signals. Dotted lines give the respective minimax bounds; triangles mark estimator-dependent sharpness checks at $t=3T$. The noise signals are those specified above and are identical for both schedules.}
    \label{fig:ijc-four-state-trajectories}
\end{figure}

Figure~\ref{fig:ijc-four-state-trajectories} shows both reconstructions and their error norms. The state and estimate curves nearly coincide at the plotted scale, while the lower panels resolve the errors. The maximum sinusoidal-noise errors are $1.0713\times10^{-3}$ and $0.7846\times10^{-3}$, respectively. For the broader-band signal, the maxima are $0.4912\times10^{-3}$ and $0.5042\times10^{-3}$. Thus the five-sample estimator improves the worst-case guarantee without implying a smaller error for every individual disturbance. Both bounds are attained in the separate sharpness check.

%% file: sections/11_discussion.tex
\section{Discussion and conclusions}
\label{sec:discussion}
\label{sec:conclusion}

The main distinction is between selecting enough measurements for exact reconstruction and selecting enough to retain all bounded-error information in a continuous history. For observable scalar-output LTV systems, the former requires $n$ samples and the latter can require $n^2$. The quadratic requirement persists for real-analytic information curves. Strict total positivity gives a stronger reduction: one $n$-point schedule attains every coordinate optimum, as does endpoint sensing in tridiagonal cooperative chains. The interior-sensing family shows that positivity and distinct real eigenvalues alone do not ensure this conclusion. Its strict gap holds over an explicit parameter range, while the nilpotent and mechanical examples show related obstructions in other system classes.

The results are pointwise in reconstruction time and depend on the chosen state norm. They count actual measurement times; they do not specify online buffer storage. For LTV systems the optimal delay set may change with time. For a fixed pattern whose information matrix is continuous and full rank on a compact time interval, the margin has a positive minimum. The pointwise error bound is therefore uniform on that interval, although the pattern need not remain optimal.

The explicit LTI families establish strict gaps of one sample; they do not imply a universal bound of $n+1$. A strict three-state gap also persists under small stabilizing damping. Consequently Hurwitz stability, tridiagonal positivity and distinct real eigenvalues do not suffice for three-sample saturation with an interior measurement. The exact count for the damped family and the case of physical RC chains with equal coordinate-error weights remain open. Determining the worst-case LTI number of measurements for saturation between $n$ and $n^2$ remains open. Another question is which LTV systems admit a fixed relative-delay pattern that saturates the history uniformly in time.

The seven-state RC chain illustrates the reduction in physical nodal coordinates. It also distinguishes sample count, maximum delay and support span: seven samples continue to suffice as the window grows, and their span can remain fixed while their location and reconstruction gain change. The trajectory comparisons illustrate the meaning of the minimax guarantee. A smaller worst-case bound does not impose an ordering of errors for every individual disturbance. The sample-count results assume exact models and measurement times. Extending the structural classification to uncertain models or constrained timing is a separate question.

%% file: sections/appendix_history_compression.tex
\section{Proof of finite temporal-history saturation}
\label{app:history-compression}

\subsection{The upper bound}

Fix the reconstruction time and identify the information rows with vectors $h(\tau)\in\R^n$. Let $Q=\operatorname{conv}\{\pm h(\tau):\tau\in[0,T]\}$ and, for a finite schedule, $Q_\tau=\operatorname{conv}\{\pm h(\tau_i)\}$. Their support functions are the continuous and sampled output norms. Since $h_{B_1^n}(x)=\|x\|_\infty$, convex set inclusion gives
\begin{equation}
    \mu_c=\max\{\rho:\rho B_1^n\subseteq Q\},
    \qquad
    \mu_\infty(\tau)=\max\{\rho:\rho B_1^n\subseteq Q_\tau\}.
    \label{eq:app-cont-crosspoly}
\end{equation}

For each $j$, the maximal radial point $q_j=\rho_j e_j\in Q$ has $\rho_j\ge\mu_c$ and lies in an exposed face of dimension at most $n-1$. In any convex representation of $q_j$ by signed temporal generators, every generator with positive weight lies in that face: equality in the supporting linear functional forces this. Carath\'eodory's theorem within the face therefore uses at most $n$ generators. Collecting their times over all $n$ coordinates gives a schedule with at most $n^2$ times. Its symmetric hull contains $\pm q_j$ and thus every $\pm\mu_c e_j$, so it contains $\mu_c B_1^n$ and attains $\mu_c$. The reverse inequality follows from $Q_\tau\subseteq Q$. Finally, the continuous rows span $\R^n$, so $n$ independent rows can be selected and fewer than $n$ cannot give full rank. This proves $m_{\mathrm{obs}}=n\le m_{\mathrm{sat}}\le n^2$.

\subsection{Analytic sharpness}

Fix $n\ge2$ and $\mu>0$. For each $j$, choose $n$ affinely independent vectors $w_{jk}\in e_j^\perp$ with barycenter zero and $\|w_{jk}\|_\infty\le1$. For $0<\varepsilon<\mu$, set
\begin{equation}
    v_{jk}=\mu e_j+\varepsilon w_{jk},
    \qquad
    Q=\operatorname{conv}\{\pm v_{jk}:j,k=1,\ldots,n\}.
    \label{eq:app-sharp-vjk}
\end{equation}
The exposed face in direction $e_j$ is exactly
\begin{equation}
    F_j=\operatorname{conv}\{v_{j1},\ldots,v_{jn}\}:
    \label{eq:app-sharp-face}
\end{equation}
its vertices have $j$th coordinate $\mu$, while all other signed generators have smaller $j$th coordinate. Its barycenter is $\mu e_j$, with all barycentric weights $1/n$. Hence $\mu B_1^n\subseteq Q$ and $h_Q(e_j)=\mu$, giving margin $\mu$.

Enumerate the $N=n^2$ vertices as $v_1,\ldots,v_N$ and choose distinct nodes $s_a\in(0,T)$. Define
\begin{equation}
    p_a(\tau)=\prod_{b\ne a}(\tau-s_b)^2,
    \qquad
    \lambda_a(\tau)=\frac{p_a(\tau)}{\sum_b p_b(\tau)},
    \qquad
    h(\tau)=\sum_a\lambda_a(\tau)v_a.
\end{equation}
The denominator is positive on $\R$: all terms are positive away from the nodes, and $p_a(s_a)>0$. Thus the curve is real analytic. Its nonnegative weights sum to one and satisfy $\lambda_a(s_b)=\delta_{ab}$, so $h(s_a)=v_a$ and its signed convex hull is $Q$. Away from the nodes every weight is positive, including weights on vertices outside each $F_j$. Consequently $-\mu<e_j^Th(\tau)<\mu$ there. The only signed curve points on $F_j$ are its $n$ prescribed vertices.

A schedule attaining margin $\mu$ must have $\mu e_j$ in its sampled hull. Equality in the supporting functional restricts any representation to $F_j$. Since $\mu e_j$ lies in the relative interior of the simplex $F_j$, all its $n$ vertices are necessary. Their nodes are disjoint across $j$, forcing all $n^2$ times. Sampling those times suffices. Finally, $A(s)\equiv0$ and $C(s)=h(t_\star-s)^T$ realize this curve as a real-analytic scalar-output LTV system at reconstruction time $t_\star$. This proves the sharpness assertion of Theorem~\ref{thm:n2-saturation}.

%% file: sections/appendix_lti_gaps.tex
\section{Proofs of the LTI structural results}
\label{app:lti-gaps}

\subsection{Endpoint sensing in tridiagonal chains}
\label{app:endpoint-chain}

It suffices to consider $C=e_1^T$; reversing the coordinates treats $e_n^T$. Let $P=\operatorname{diag}(1,-1,1,\ldots)$ and $B=-PAP$. Then $B$ is tridiagonal cooperative and
\begin{equation}
    Ce^{-A\tau}P=e_1^Te^{B\tau}=:h(\tau).
\end{equation}
The column sign change preserves every coordinate recovery constant. We prove the result for $h$.

For $k\le n$ set $R_k=\operatorname{col}_{r=0}^{k-1}e_1^TB^r$. Tridiagonality implies that its only possibly nonzero $k$-column minor uses columns $1,\ldots,k$, and
\begin{equation}
    \det R_k[:,1{:}k]=\prod_{r=1}^{k-1}b_{r,r+1}^{\,k-r}>0.
\end{equation}
For any ordered $k$-subset $J$ of coordinates, Cauchy--Binet therefore gives the temporal Wronskian identity
\begin{equation}
    W_J(\tau)=\left(\prod_{r=1}^{k-1}b_{r,r+1}^{\,k-r}\right)
        \det(e^{B\tau}[1{:}k,J]).
    \label{eq:endpoint-wronskian}
\end{equation}
The classical tridiagonal total-positivity theorem \cite{margaliot2019tpds} makes the right-hand side positive for $\tau>0$. The Wronskian criterion for Chebyshev systems then implies positivity of every ordered temporal collocation minor at positive times. At zero the minors are nonnegative by continuity. Every square collocation matrix at distinct nonnegative times remains nonsingular: if the first time is zero, expansion along $h(0)=e_1^T$ leaves a strictly positive $(n-1)$-minor. Its inverse thus has the nonnegative checkerboard sign pattern.

The leading Wronskians for $J=\{1,\ldots,k\}$ are positive also at zero by \eqref{eq:endpoint-wronskian}. Hence $h_1,\ldots,h_{n-1}$ form a Haar space on $[0,T]$. The best uniform approximation of $h_n$ from that space has a normalized residual $q(\tau)=h(\tau)z$ with $n$ alternating unit contacts. If $O$ is the corresponding collocation matrix and $s$ is the alternating contact vector, then $Oz=s$. Its inverse signs give
\begin{equation}
    \alpha_j(\tau^\star)=\|\operatorname{row}_j(O^{-1})\|_1=|z_j|
    \le\alpha_j^{\mathrm{cont}}\le\alpha_j(\tau^\star).
\end{equation}
Both inequalities are equalities. This proves simultaneous coordinatewise attainment, including a possible zero-delay contact, and hence Corollary~\ref{cor:endpoint-chain}.

\subsection{The three-state interior-sensing family}
\label{app:interior-gap}

Set $P=\operatorname{diag}(-1,1,-1)$ and $u=e^{\omega\tau}$. Direct calculation for \eqref{eq:interior-chain-model} yields
\begin{equation}
    G(u):=Ce^{-A_w\tau}P
    =\begin{bmatrix}
    \dfrac{3(u^2+2u-3)}{16}&
    \dfrac{3u^2+2u+3}{8}&
    \dfrac{3u^2-2u-1}{16w}
    \end{bmatrix},\quad 1\le u\le R.
    \label{eq:interior-quadratic-row}
\end{equation}
For example, the polynomial coefficients satisfy $G(1)=e_2^T$ and $uG'(u)=G(u)(-PA_wP/\omega)$, which proves the identity by uniqueness of the matrix differential equation. Column signs and the monotone change from $\tau$ to $u$ preserve the sampling problem. The coefficient matrix is invertible, so any three distinct nodes give observability.

If $p(u)=G(u)z=c_0+c_1u+c_2u^2$, then
\begin{equation}
    z_1=-c_0+c_1+\tfrac13c_2,\qquad
    z_2=c_0+c_1+c_2,\qquad
    z_3=w(-c_0-3c_1+3c_2).
    \label{eq:interior-coordinate-functionals}
\end{equation}
Write $m=(R+1)/2$ and $d=(R-1)^2$. Quadratic interpolation at $1,m,R$ represents $z_1$ with coefficients
\begin{equation}
    \left(-\frac{3R^2+12R+1}{3d},\quad
    \frac{8(3R+1)}{3d},\quad
    -\frac{6R+10}{3d}\right).
    \label{eq:interior-first-row}
\end{equation}
Their norm is $a(R)$. The polynomial
\begin{equation}
    q(u)=1-2\left(\frac{2u-R-1}{R-1}\right)^2
\end{equation}
satisfies $|q(u)|\le1$, with equality precisely at $1,m,R$ and signs $(-,+,-)$. Substitution in \eqref{eq:interior-coordinate-functionals} gives $z_1=a(R)$. Thus the continuous first-coordinate optimum is $a(R)$. All coefficients in \eqref{eq:interior-first-row} are nonzero. Primal--dual equality forces every optimal first-coordinate representation onto these three contacts, and independence makes the representation unique. Every schedule attaining the common value $a(R)$ must therefore contain all three nodes.

The second-coordinate optimum is $1$: measurement at $u=1$ recovers $z_2$ directly, and the continuously feasible polynomial $p\equiv1$ attains $z_2=1$.

For the third coordinate let
\begin{equation}
    s=\frac{3(R+1)}{R-3},\qquad D=R^2-6R-3.
\end{equation}
For $R>9$, one has $1<s<m<R$ and $D>0$. The exact representation
\begin{equation}
    z_3=w\left(-\frac{(R-3)^2}{D}p(s)+\frac{12}{D}p(R)\right)
    \label{eq:interior-third-two-nodes}
\end{equation}
follows by checking the constant, linear and quadratic coefficients. Its norm is $wc(R)$. The matching feasible dual polynomial
\begin{equation}
    p_3(u)=-1+2\left(\frac{u-s}{R-s}\right)^2
\end{equation}
has values $-1,+1$ at $s,R$ and satisfies $|p_3(u)|\le1$ on $[1,R]$ because $s<m$. Hence $\alpha_3^{\mathrm{cont}}=wc(R)$, and \eqref{eq:interior-gap-range} gives the common continuous gain $a(R)$.

At the mandatory three nodes $1,m,R$, the unique third-coordinate interpolation row is
\begin{equation}
    \frac{w}{d}\left(-(R-9)(R+1),\ -8(R+3),\ 2(R+7)\right).
\end{equation}
Its norm is $wb(R)>a(R)$. No three-node schedule can therefore attain the continuous common gain. Adding $s$ retains the optimal first-coordinate representation and the direct second-coordinate measurement, while supplying \eqref{eq:interior-third-two-nodes}. These four nodes attain all three continuous coordinate optima. Finally,
\begin{equation}
    b(R)-c(R)=\frac{4(R-9)(R+1)(R+3)}{(R-1)^2(R^2-6R-3)}>0
\end{equation}
proves that the parameter interval is nonempty.

For $R=16,w=7/4$, the four-node matrix has rows $G(u)$ at $u=(1,51/13,17/2,16)$. An exact left inverse is
\begin{equation}
    L_4=\begin{bmatrix}
    -961/675&0&392/675&-106/675\\
    1&0&0&0\\
    0&-1183/628&0&21/157
    \end{bmatrix}.
    \label{eq:interior-exact-left-inverse}
\end{equation}
Its row norms are $1459/675$, $1$ and $1267/628$. For the physical delayed rows, $PL_4$ is the left inverse and has the same norms.

To prove Corollary~\ref{cor:hurwitz-interior-gap}, let $\mu_c$ be the undamped continuous margin and $\mu_3^\star$ the maximum margin over the compact delay cube $[0,T]^3$. Theorem~\ref{thm:interior-gap} excludes equality at every three-delay schedule, including singular schedules with repeated times. Continuity therefore gives $\Delta=\mu_c-\mu_3^\star>0$. The damped information row is $h_\varepsilon(\tau)=e^{\varepsilon\tau}h_0(\tau)$, so
\begin{equation}
    \sup_{\tau\in[0,T]}\|h_\varepsilon(\tau)-h_0(\tau)\|_1
    \le (e^{\varepsilon T}-1)M,
    \qquad M=\sup_{\tau\in[0,T]}\|h_0(\tau)\|_1.
\end{equation}
The minimum gain is one-Lipschitz in this row norm, uniformly over schedules; the same estimate applies to the continuous margin. Their gap is thus at least $\Delta-2(e^{\varepsilon T}-1)M$, which is positive for sufficiently small $\varepsilon>0$. Scalar multiplication of each row preserves sampled rank, and the eigenvalues become $-\varepsilon,-\omega-\varepsilon,-2\omega-\varepsilon$. This proves the corollary.

\subsection{Exact proof of the four-state gap}
\label{app:exact-gap-proof}

Work in the forward variable $t=\tau-1$ of \eqref{eq:counterexample-delay-shift}. For the nodes
\begin{equation}
    t_5=\left(-1,-\frac35,\frac2{17},\frac35,1\right),
\end{equation}
let $O_5$ have rows $H(t_i)$. Direct multiplication gives $L_5O_5=I$, where
\begin{equation}
    L_5=\begin{bmatrix}
    -\frac{349}{1280}&\frac{95}{768}&0&-\frac5{768}&\frac{39}{1280}\\[1mm]
    -\frac{157}{640}&\frac{45}{128}&0&-\frac{55}{128}&\frac{47}{640}\\[1mm]
    -\frac{783}{1520}&0&\frac{4913}{11400}&0&-\frac{349}{1200}\\[1mm]
    \frac{93}{256}&-\frac{475}{768}&0&\frac{25}{768}&-\frac{39}{256}
    \end{bmatrix}.
    \label{eq:exact-gap-left-inverse}
\end{equation}
Its row norms are $13/30$, $11/10$, $14101/11400$, and $7/6$.

To obtain the matching dual lower bound, set
\begin{equation}
    p_\star(t)
    =\frac{78881+39304t-160106t^2-39304t^3}{81225}.
\end{equation}
The factorizations
\begin{equation}
\begin{aligned}
    p_\star(t)-1&=-\frac{2(17t-2)^2(68t+293)}{81225},\\
    p_\star(t)+1&=-\frac{578(t-1)(t+1)(68t+277)}{81225}
\end{aligned}
\end{equation}
prove $|p_\star(t)|\le1$ on $[-1,1]$, with equality only at $-1$, $2/17$, and $1$. If $c_\star$ is its coefficient vector and $z_\star=M^{-1}c_\star$, then $H(t)z_\star=p_\star(t)$ and
\begin{equation}
    (z_\star)_3=\frac{14101}{11400}=: \alpha_\star.
\end{equation}
Thus the dual lower bound equals the third row norm in \eqref{eq:exact-gap-left-inverse}. This proves the continuous optimum, five-sample attainment, and strict noncritical slack.

If a finite schedule attains this optimum, a coordinate-$3$ reconstruction row $a$ must satisfy
\begin{equation}
    \alpha_\star
    =\sum_i a_i p_\star(t_i)
    \le\sum_i|a_i|
    \le\alpha_\star.
\end{equation}
Equality forces every nonzero coefficient onto the three contacts. Their evaluation rows are independent, and the third row of \eqref{eq:exact-gap-left-inverse} is the unique representation supported there. All three coefficients are nonzero, so all three contacts are mandatory.

A potentially saturating four-point schedule therefore has nodes $(-1,2/17,1,s)$, with $-1<s<1$ and $s\ne2/17$. The unique inverse gives, for its coordinate-$2$ row norm,
\begin{equation}
    \rho_2(s)-\frac32=\begin{cases}
    \displaystyle\frac{5P_1(s)}{4(s-1)(s+1)(17s-2)},&-1<s<3/103,\\[2mm]
    \displaystyle-\frac{P_2(s)}{150(s+1)(17s-2)},&3/103\le s<2/17,\\[2mm]
    \displaystyle-\frac{P_3(s)}{95(s-1)(17s-2)},&2/17<s<1,
    \end{cases}
\end{equation}
where
\begin{equation}
\begin{aligned}
    P_1(s)&=-17s^3+2s^2+17s+6,\\
    P_2(s)&=4063s^2+4435s+1122,\\
    P_3(s)&=2737s^2-3484s+1222.
\end{aligned}
\end{equation}
The first numerator is positive on its interval because
\begin{equation}
    P_1(s)=(2-17s)(s+\tfrac12)^2
        +17(s+\tfrac{77}{136})^2+\frac{55}{1088}.
\end{equation}
The second is positive for $s>0$, and the third has positive leading coefficient and discriminant $-1240200<0$. The denominators have signs $+,-,-$ on their respective intervals. Thus $\rho_2(s)>3/2>\alpha_\star$ for every admissible fourth node, excluding saturation. Any four distinct nodes have full rank because $\det M\ne0$ and the polynomial collocation matrix is Vandermonde. Together with \eqref{eq:exact-gap-left-inverse}, this proves $m_{\mathrm{obs}}=4$ and $m_{\mathrm{sat}}=5$.

\subsection{Computer-assisted proof for the mechanical system}
\label{app:building-proof}

We prove Proposition~\ref{prop:building-gap} by outward-rounded interval inclusions establishing the continuous optimum, six-sample exclusion, and seven-sample attainment.

We use $H(t)=Ce^{At}$ and transfer the result to delayed reconstruction by \eqref{eq:building-time-reversal}. The modal representation gives explicit expressions suitable for interval evaluation. For $k_r\in\{1,3,5\}$,
\begin{equation}
    \omega_r=2\sin\frac{k_r\pi}{14},\qquad
    V_{jr}=\frac2{\sqrt7}\sin\frac{jk_r\pi}{7},
\end{equation}
and, for $j=1,2,3$,
\begin{equation}
    H_j(t)=\sum_{r=1}^3V_{1r}V_{jr}\cos(\omega_rt),
    \qquad
    H_{j+3}(t)=\sum_{r=1}^3V_{1r}V_{jr}
        \frac{\sin(\omega_rt)}{\omega_r}.
\end{equation}
The six modal sine and cosine functions are independent, $V$ is orthogonal, and every $V_{1r}$ is nonzero. The temporal row family therefore spans $\mathbb R^6$, giving $m_{\mathrm{obs}}=6$.

\paragraph{Critical representation and dual feasibility.}
Let $\sigma=(1,-1,-1)$. The twelve unknowns $\theta=(t_1,t_2,t_3,a_1,a_2,a_3,z)$ solve
\begin{equation}
    \sum_{i=1}^3 a_iH(t_i)^T=e_3,
    \qquad H(t_i)z=\sigma_i,
    \qquad H'(t_i)z=0\quad(i=1,2,3).
    \label{eq:building-kkt-system}
\end{equation}
Write these equations as $F(\theta)=0$. For a box $X$ about a numerical centre $\theta_0$ and a nonsingular point preconditioner $R$, interval evaluation establishes the strict Krawczyk inclusion
\begin{equation}
    \theta_0-RF(\theta_0)
    +(I-R[DF](X))(X-\theta_0)
    \subset\operatorname{int}X,
\end{equation}
which encloses a unique root in $X$. Interval bounds also give $a_1>0$, $a_2<0$, and $a_3<0$ throughout $X$. The resulting contact enclosures, widened for decimal display, are
\begin{equation}
\begin{aligned}
    \tau_{c,1}&\in[1.848607960412,\,1.848607960415],\\
    \tau_{c,2}&\in[5.620511863042,\,5.620511863045],\\
    \tau_{c,3}&\in[7.986151862406,\,7.986151862409].
\end{aligned}
\end{equation}
The sum $\alpha=\sum_i|a_i|$ lies in \eqref{eq:building-alpha}.

For $g(t)=H(t)z$, derivative tests isolate exactly five stationary points on $[0,10]$. On each isolating interval, $g''$ has fixed sign and $g'$ changes sign between its endpoints; interval subdivision excludes zeros of $g'$ elsewhere. At the three contact extrema, \eqref{eq:building-kkt-system} gives the exact values $g(t_i)=\sigma_i$, and the signs of $g''$ give strict maxima or minima. The other two extrema and the endpoints satisfy $|g|<1$. Hence $|g(t)|\le1$ everywhere, with equality only at the three contacts. Primal--dual equality yields
\begin{equation}
    z_3=\sum_i a_i g(t_i)=\sum_i|a_i|=\alpha,
\end{equation}
so $\alpha_3^{\mathrm{cont}}=\alpha$. A determinant enclosure excluding zero for the $3\times3$ minor of the contact rows in columns $1,3,5$ proves their independence. Since every $a_i$ is nonzero, the critical representation is unique and each contact is mandatory in a saturating schedule.

\paragraph{Conditional exclusion of six samples.}
It remains to exclude schedules consisting of the three contacts and three free times $0\le s_1\le s_2\le s_3\le10$. The ordered domain is covered initially by the 220 boxes
\begin{equation}
    [i,i+1]\times[j,j+1]\times[k,k+1],
    \qquad 0\le i\le j\le k\le9.
\end{equation}
For a proposed dual vector $w=(q,v)$, the conserved energy
\begin{equation}
    E(w)=q^TKq+v^Tv
\end{equation}
gives $|H'(t)w|\le\sqrt{E(w)}$. Thus on an interval $I=[c-r,c+r]$,
\begin{equation}
    \sup_{t\in I}|H(t)w|
    \le |H(c)w|+r\sqrt{E(w)}.
    \label{eq:building-box-bound}
\end{equation}
For each box, interval arithmetic bounds all three free-time intervals and all three contact enclosures by a common upper bound $\beta_B$. If a positive coordinate satisfies $w_j/\beta_B>\alpha^{\mathrm{upper}}$, then $w/\beta_B$ is dual feasible throughout that box and excludes saturation there. Floating-point linear programs only propose $w$; its displayed decimal coordinates define the witness, and the interval inequality decides acceptance. Boxes for which this inequality is not verified are bisected.

A finite subdivision into 4049 terminal boxes covers the entire ordered domain. On every box, the verified dual lower bound exceeds $2.8056$, while $\alpha^{\mathrm{upper}}<2.79548$. Therefore no six-sample schedule containing the three mandatory contacts attains the continuous optimum. Since every saturating schedule must contain those contacts, saturation is excluded for all six-sample schedules. The quantitative bound $2.8056$ applies only to the contact-constrained domain. The unrestricted conclusion is
\begin{equation}
    \max_{\tau\in[0,10]^6}\mu_\infty^-(\tau)
    <\mu_{\infty,\mathrm{delay}}^{\mathrm{cont}}(10):
\end{equation}
the maximum exists by continuity on the compact delay cube, and equality has been excluded.

\paragraph{Seven-sample attainment.}
For the seven exact nodes in \eqref{eq:building-seven-schedule}, coordinate $3$ uses the three KKT weights. For the remaining coordinates, interval linear solves of $O_{\mathcal S_j}^Ta_j=e_j$ use the following six-node supports; indices refer to the ordered seven-node schedule:
\begin{equation}
\begin{aligned}
    \mathcal S_1&=\{1,2,3,4,5,7\},&
    \mathcal S_2&=\{1,2,4,5,6,7\},\\
    \mathcal S_4=\mathcal S_5&=\{1,3,4,5,6,7\},&
    \mathcal S_6&=\{1,2,4,5,6,7\}.
\end{aligned}
\end{equation}
Rounded upward, the corresponding row norm bounds are $2.26224$, $2.18332$, $2.62589$, $2.74306$, and $1.78400$. All lie below the lower bound on $\alpha$. Thus coordinate $3$ is uniquely critical for the continuous problem, the seven-node schedule attains its optimum, and the six-sample exclusion proves $m_{\mathrm{sat}}=7$.